\documentclass{elsarticle}
\usepackage{amsmath,amssymb,graphicx, amsthm}
\newtheorem{Thm}{Theorem}
\newtheorem{Prop}{Proposition}
\newtheorem{Def}{Definition}
\newtheorem{Cor}{Corollary}
\newtheorem{Example}{Example}
\newtheorem{Rmk}{Remark}

\begin{document}

\begin{frontmatter}
\journal{}

\title{A secure solution on hierarchical access control }
 \author[Wei-Huang]{Chuan-Sheng Wei}\ead{cswei@fcu.edu.tw}
  \author[Chen]{Sheng-Gwo Chen\corref{corresponding}}
 \author[Wei-Huang]{Tone-Yau Huang}\ead{huangty@fcu.edu.tw}
 \author[Ong]{Yao Lin Ong}\ead{ylong@mail.cjcu.edu.tw}

\cortext[corresponding]{Corresponding author email :
csg@mail.ncyu.edu.tw}
\address[Wei-Huang]{Department of Applied Mathematics , Feng Chia University, Taichung 407, Taiwan.  }

\address[Chen]{Department of Applied Mathematics, National Chiayi University,  Chia-Yi 600,
Taiwan.}
\address[Ong]{Department of Accounting and Information System, Chang Jung Christian University, Tainan 711, Taiwan}

\begin{abstract}
Hierarchical access control  is an important and traditional problem
in information security. In 2001, Wu et.al. proposed an elegant
solution for hierarchical access control by the secure-filter. Jeng
and Wang presented an improvement of Wu et. al.'s method by the ECC
cryptosystem. However, secure-filter method is insecure in dynaminc
access control. Lie, Hsu and Tripathy, Paul pointed out some secure
leaks on the secure-filter and presented some improvements to
eliminate these secure flaws. In this paper, we revise the
secure-filter in Jeng-Wang method and propose another secure
solutions in hierarchical access control problem.  CA is a super
security class (user) in our proposed method and the secure-filter
of $u_i$ in our solutions is a polynomial of degree $n_i+1$ in
$\mathbb{Z}_p^*$, $f_i(x)=(x-h_i)(x-a_1)\cdots(x-a_{n_i})
+L_{l_i}(K_i)$. Although the degree of our secure-filter is larger
than others solutions, our solution is secure and efficient in
dynamics access control.

\end{abstract}

\end{frontmatter}

\section{Introduction}
The access control problems in a computer (communication) system are
structured as a user hierarchy. Under such a hierarchy, the users
and their own data are organized into a number of disjoint set of
security classes and each user is assigned to a security class
called the user's security clearance. Let $S=\{ u_1, u_2, \cdots,
u_n\}$ be a set of $n$-disjoint security classes. Assume that
``$\leq$'' is a partial ordering on $S$. In $(S, \leq)$, $u_j \leq
u_i$ means that the user $u_i$ has a security clearance higher than
or equal to the user $u_j$ in the security class $(S,\leq)$. In
other words, the user $u_i$ can read the information items which
belong to the user $u_j$, but $u_j$ cannot read the data which
belong to $u_i$ in this partial ordering system. $u_j$ is called the
successor of $u_i$ and $u_i$ is called the predecessor of $u_j$ if
$u_j \leq u_i$.

The first solution of hierarchical access control problem is
developed by  Akl and Taylor\cite{Akl} in 1983. They assigned a
public integer $t_i$ to each security user $u_i$ with the property
that $t_i | t_j$ iff $u_j \leq u_i$. The central authority (CA) in
their scheme assigned a security number $K_0$ and the cryptographic
key for $u_i$ is computed by $K_i = K_0^{t_i}  (\textrm{mod } p)$
where $p$ is a large prime number. The predecessor $u_i$ of $u_j$
can obtain the cryptographical key $K_j$ by
\begin{equation}
K_j = (K_0^{t_i} )^{\frac{t_j}{t_i}} = (K_i)^{\frac{t_j}{t_i}} ( \textrm{mod } p).
\end{equation}
Obviously, the user $u_i$ can obtain $K_j$ for some user $u_j$ if
and only if $u_i$ is a predecessor of $u_j$, that is $u_j \leq u_i$.
However, the size of the public information $t_i$ will grow linearly
with the number of security classes in Alk and Taylor's scheme.
Hence, MacKinnon et.al.\cite{MacKinnon} and Harn et.al.\cite{Harn}
presented different key assignment schemes to improve Akl and
Taylor's method, respectively. However, all secret keys must be
re-generated when a security class is added into or deleted from the
user hierarchy in these methods. The dynamic access control problems
 cannot be efficiently solved.

Since then, several solutions\cite{Chang1,Chang2,Chick} have been
proposed to archive the target of the dynamic access control problem
in hierarchy by keeping the size of public information as small as
possible. However, these approaches are inefficient if the user with
a high security of his successor without an immediate one. In 1995,
Tsai and Chang\cite{Tsai} presented an elegant key management scheme
based on the Rabin public cryptosystem and Chinese remainder
theorem. Their method is efficient in the key-generation and
key-derivation processes, but a group of secret data is required to
be held by each security class.

Recently,  Wu et.al.\cite{Wu} proposed a novel key management scheme
by the secure-filter function to solve the dynamic access control
problems. The secure-filter function that Wu et.al. used is defined
as
\begin{equation}
f_i(x) = \prod_{t} (x-g_i^{s_t}) + K_i
\end{equation}
for all $u_i < u_t$ where $s_t$ is the secret number of $u_i$ and
$g_i$ is a public random integer chosen by $u_i$ (see \cite{Wu} for
further details). In 2006, Jeng and Wang\cite{Jeng}  presented an
efficient solution of access control problem in hierarchy by
combining the secure-filter and ECC cryptosystem. The secure-filter
in Jeng-Wang  method  is given by
\begin{equation}
f_i(x) = \prod_t (x - \tilde{A}(n_iP_t)) +K_i
\end{equation}
for all $u_i < u_t$ where $\tilde{A}$ is a function from an elliptic
group to $\mathbb{Z}_p^*$ (see \cite{Jeng} for the details ). The
major advantages of Wu et.al's and Jeng-Wang schemes were to solve
the dynamic key management efficiently. Both of them are not
necessary to re-generate keys for all the security classes in the
hierarchy when the security class is added into or deleted from the
user hierarchy.

In 2009-2011, Lin and Hsu\cite{Lin,Lin2} pointed out a security leak
of the Jeng-Wang method. To eliminate this security flaw, Lin and
Hsu revised the secure-filter to
\begin{equation}
f_i(x) = \prod_t (x-h(r\|\tilde{A}(n_iP_t))) + K_i
\end{equation}
for all $u_i < u_t$, where $h$ is an one-way hash function and $r$
is a random number. Lin-Hsu scheme is an effective secure method to
overcome the security leak of Jeng-Wang method. However, this method
must be re-computed the secure-filter $f_i$ for each $i$ when an
user is added into or deleted from the hierarchy\cite{Lin}. Tripathy
and Paul\cite{Tripathy} also proposed another attacked method on
Jeng-Wang method in dynamic access control problem in 2011.

In this paper, we propose a novel secure-filter to overcome  Lin-Hsu
and Tripathy-Paul attacks. In our method, the CA needs to choose two
addition random integers. One is a secure number $n_i$ and the other
is a public number $h_i$. We choose the secure-filter as
\begin{equation}
f_i(x)=(x-h_i) \cdot \left (\prod_t (x - \tilde{A}(n_iP_t)) \right ) + L_{l_i}(K_j)
\end{equation}
where $L_{l_i}$ is the cyclic $l_i$-shift operator. Our proposed
methods eliminate all security flaws of secure filter in Jeng-Wang
method that we know and is efficient in dynamic access control.

In section 2, we introduce Wu et.al. and Jeng-Wang key management
schemes, the security leaks of Jeng-Wang method and Lin-Hsu
improvement. We propose our method in section 3 and discuss the
dynamic access control problem in section 4.

\section{Preliminaries}
Let $S=\{u_1,u_2,\cdots, u_n\}$ be a set with a partial ordering
``$\leq$''. $u_i \leq u_j$ means that the user $u_j$ has a security
clearance higher than or equal to $u_i$. For simplicity,  $u_i <
u_j$ means that the user $u_j$ has a security clearance higher than
$u_i$ and $u_j$ is called a strictly predecessor of $u_i$. Denote
the set of all strictly predecessors of $u_i$ by $S_i=\{u_j | u_i <
u_j , u_j \in S\}$ .
\begin{Def}\cite{Wu}
The secure-filter on $K$ with respect to $\Sigma = \{ s_i | s_i \in
\mathbb{Z}_p, 0 \leq i \leq n-1 \}$ is a polynomial $f$ with the
indeterminate $x$ over a Galois field $\textrm{GF}(p)$ such that if
$s \in \Sigma$, then $f(s) \equiv  K ( \textrm{mod }  p)$ where $p$
is a public large prime.
\end{Def}

In Jeng-Wang and Wu et.al.s'  methods,  CA  publishes these
coefficients $a_{n-1}, a_{n-2}, \cdots, a_0$  of the general form of
$f$,

\begin{equation}
f(x) = x^n + a_{n-1}x^{n-1} + \cdots + a_1x + a_0.
\end{equation} .

\subsection{ Two solutions of hierarchical access control problems}

Now, let us introduce Wu et.al and Jeng et.al.s' methods roughly.
\begin{description}
\item[Wu et. al.'s scheme\cite{Wu}]: \\
\begin{description}
\item[Key generation phase] - \\
\begin{enumerate}
\item Choose distinct secret integers $s_i \in \mathbb{Z}_p$, $ 0 \leq i \leq n-1$.
\item For each $u_i \in S$ do step 3-5
\item Determine the set $\Sigma_i = \{ g_i^{s_j} | u_j \in S_j \}$.
\item Set
\begin{equation}
f_i(x) = \left \{ \begin{array}{ll}
\begin{array}{l}
(x - g_i^{s_0})(x -g_i^{s_1})\cdots(x-g_i^{s_{N_i}}) +K_i \cr = \prod_{a \in \Sigma_i} (x-a) +K_i \end{array} & \mbox{ if } S_i \neq \phi \cr\cr
 0 & \mbox{ if } S_i = \phi
\end{array}. \right.
\end{equation}
\item Assign $s_i, g_i$ and all coefficients of the general form of $f_i$ to $u_i$; $u_i$ keeps $s_i$ and $K_i$ private, and makes $(f_i, g_i)$ public.
\end{enumerate}

\item[Key derivation phase ] -

Consequently, the predecessor $u_j$ of $u_i$ can   compute
$g_i^{s_j}$ from his secure number $s_j$ and the public
number $g_i$. Then the user $u_j$  obtain the cryptographic
key $K_i$ by $f_i(g_i^{s_j})$.

\end{description}

\item[Jeng-Wang method\cite{Jeng}]: \\
\begin{description}
\item[Initialization phase ]- \\
\begin{enumerate}
\item CA randomly choose a large prime $p$, and $a,b \in
  \mathbb{Z}_p$  with  $4a^3 + 27b^2 \textrm{ mod } p
  \neq 0$. There is a group of elliptic curve $E_p(a,b)$
over a Galois field $GF(p)$ that contains a set of
points $(x,y)$ with $x,y \in \mathbb{Z}^*_p$ satisfying
$y^2 = x^3+ax+b (\mbox{mod } p)$ and a point $O$ at
infinity. Let $G \in E_p(a,b)$ be a base point  whose
order is a large prime $q$.
\item CA selects an algorithm $\tilde{A} :(x,y) \rightarrow v$ for representing a point $(x,y)$ on $E_p(a,b)$ as a real number $v$.
\item CA publishes $(p,q,\tilde{A},E,G)$.
\end{enumerate}

\item[Key generation phase ] - \\
\begin{enumerate}
\item CA determines its security key $n_{ca} \in \mathbb{Z}_q$ and publishes the public key $P_{ca} =  n_{ca}G$.
\item Each security class $u_i$ chooses its encryption key $K_i$, a secret key $n_i \in \mathbb{Z}_q$ and a public key $P_i =n_iG$.
\item $u_i$ chooses a random integer $k$ and encrypts $(K_i,n_i)$ as a ciphertext $\{kG, (K_i,n_i)+kP_{ca} \}$, and sends this ciphertext to CA.
\item CA decrypts $\{kG, (K_i, n_i)+kP_{ca} \}$ to obtain $(K_i, n_i)$ by the equation
\begin{equation}
(K_i,n_i) = ((K_i,n_i)+kP_{ca}) - n_{ca}(kG).
\end{equation}
\item CA generate the secure-filter
\begin{equation}
f_i(x) = \left \{ \begin{array}{ll}
\begin{array}{l}
(x -\tilde{A}(n_1P_i)) \cdots(x-\tilde{A}(n_{N_i}P_i)) +K_i \cr
 = \prod_{u_j \in S_i} (x-\tilde{A}(n_jP_i)) +K_i \end{array} & \mbox{ if } S_i \neq \phi \cr\cr
 0 & \mbox{ if } S_i = \phi
\end{array}. \right.
\end{equation}

\end{enumerate}

\item[key derivation phase] -

When $u_j$ is a predecessor of $u_i$, $u_j$ can use the
secure key $n_j$, the public key $P_i$ and the public
secure-filter $f_i$ to derive $K_i =
f_i(\tilde{A}(n_jP_i))$.

\end{description}

\end{description}

\subsection{Two attacks on Jeng-Wang method}
Both of above solutions are secure in a static hierarchical system,
it means that no user is added into or deleted from $S$. However,
these schemes are insecure when a user joins or removes from the
hierarchy. Two different attacks on Jeng-Wang method have be
presented. We discuss these attacks on the Jeng-Wang method (and the
Wu et. al.'s method) in this subsection. For simplicity, we denote
the secure-filter $f_i(x)$ of a security class $u_i$ as
\begin{equation}\label{secure_filter_f}
f_i(x) = (x-c_{i,1})(x-c_{i,2})\cdots(x-c_{i,N_i}) + K_i.
\end{equation}
The value  $c_{i,j}$ plays the same role as $g_i^{s_j}$ in Wang et.
al.'s scheme or
 $\tilde{A}(n_jP_i)$ in Jeng-Wang method.

\begin{description}
\item[Lin-Hsu attack\cite{Lin,Lin2}] : \\
Without lose of  the generality, we only consider the situation
that there is a new security class joining in hierarchy. Let
$u_i$ be a security class (user) in $S$, $S_i =\{ u_j | u_i <
u_j \mbox{ and } u_j \in S\}$ be the set of strictly
predecessors of $u_i$ and $u_N$ be a new predecessor of $u_i$.
Assume that
$$f_i(x)=(x-c_{i,1})(x-c_{i,2})\cdots(x-c_{i,N_v}) + K_i $$ is the secure-filter of $u_i$ in the original hierarchy
and
$$\tilde{f}_i(x) = (x-c_{i,1})(x-c_{i,2})\cdots(x-c_{i,N_v})(x-c_{i,N}) + K_i $$
is the new secure-filter in the new hierarchy  after $u_N$ is
inserted.

One has
\begin{equation}
\tilde{f}_i(x) - f_i(x) = (x - c_{i,1})(x-c_{i,2})\cdots(x - c_{i,N_i})(x-c_{i,N}-1).
\end{equation}
Obviously, $c_{i,1}\cdots c_{i,N_i}$ are the solutions of
$\tilde{f}_i(x)-f_i(x) = 0$. The adversary can obtain $c_{i,j}$
for some $j \in \{1,2,\cdots,N_i\}$ by solving the equation $
(\tilde{f}_i - f_i)(x)=0$, and  then get the encryption key
$K_i$ of $u_i$ by $f_i(c_{i,j})$ or $\tilde{f}_i(c_{i,j})$.

\item[Tripathy-Paul attack\cite{Tripathy}] : \\
In 2011, Tripathy and Paul pointed out another attack on
Jeng-Wang method. With the assumptions in Lin-Hsu attack, the
general form of the secure-filters $f_i(x)$ and $\tilde{f}_i(x)$
are
\begin{equation}
f_i(x) = x^{N_i} + \alpha_{N_i -1}x^{N_i-1} + \cdots + \alpha_1x^1 + \alpha_0
\end{equation}
and
\begin{equation}
\tilde{f}_i(x) = x^{N_i+1}+\beta_{N_i}x^{N_i} + \cdots + \beta_1x^1+\beta_0,
\end{equation}
respectively. Since
\begin{equation}
\begin{array}{ll}
\alpha_{N_i-1} & = c_{i,1}+c_{i,2}+\cdots + c_{i,N_i}, \cr
 \beta_{N_i} & =  c_{i,1}+c_{i,2}+\cdots +c_{i,N_i} + c_{i,N},
 \end{array}
 \end{equation}  and $\alpha_j$,
$\beta_k$ for each $j,k$ are public information, the adversary
can obtain $c_{i,N}$ by subtracting $\beta_{N_i}$  from
$\alpha_{N_i-1}$, and the encryption key $K_i =
\tilde{f}_i(c_{i,N})$.
\end{description}

\subsection{ Lin-Hsu improvement}
  Lin and Hsu \cite{Lin,Lin2} presented an
improvement of Jeng-Wang scheme. They used a random number $r$ and
an one-way hash function $h(\cdot)$ replacing $(\tilde{A}(n_jP_i),
K_i)$ in Jeng-Wang scheme with $(h(r\|\tilde{A}(n_jP_i)), K_i)$.
That is, the value $c_{i,j}$ in Eq.(\ref{secure_filter_f}) for each
$j$ should be changed when a security class is added in or deleted
from in hierarchy. This implies that $\tilde{A}(n_jP_i)$ is not
solution of $\tilde{f}_i(x)-f_i(x)=0$ anymore. Hence, this method
eliminates the security flaw on Jeng-Wang method effectively (see
\cite{Lin,Lin2} for the details).

\section{Our proposed method}

First, let us introduce some definitions and   properties of the
cyclic $l$-shift operator.

\begin{Def}
Let $b$ be a positive integer greater than 1. The radix (base) $b$
expansion of a positive integer $k$ is a unique expansion of $k$ as
\begin{equation}
k=k_{m+1} b^m+ k_{m}b^{m-1}+ \cdots +k_2b +k_1
\end{equation}
where $ k_i \in \{0,1,\cdots, b-1 \}$ for each $i$ and $k_{m+1} \neq
0$. This expansion is written as $(k_{m+1}k_{m}\cdots k_1)_b$. The
integer $k_i$ is called the i-th coefficient of the radix $b$
expansion of $k$, whereas  $k_{m+1}$ is   the leading coefficient of
the radix $b$ expansion of $k$.
\end{Def}

\begin{Def}
Let $b$ be a positive integer greater than 1 and $l$ be an integer.
The cyclic $l$-shift operator $L_l(\cdot)$  based on $b$ is defined
by
\begin{equation}
   L_l(k) = (k_{m+1}k_{\sigma^l(m)}k_{\sigma^l(m-1)}\cdots k_{\sigma^l(1)})_b = k_{m+1}b^m + \sum_{i=1}^m k_{\sigma^l(i)}b^{i-1},
\end{equation}
where $(k_{m+1}k_m\cdots k_1)_b$ is the radix $b$ expansion of $n$
and $\sigma =(1,2,\cdots,m)$ is the m-cyclic in the symmetric group
$S_m$ of degree $m$.

\end{Def}
\begin{Example}
If $b=10$, then $ L_1(21349) = 23491$, $L_2(21349) = 24913$, $L_3(21349) = 29134$ and $L_4(21349)=21349$. \\
If $b=2$, then $L_1( (11110)_2 ) = (11101)_2$, $L_2( (11110)_2 ) =
(11011)$, $L_3((11110)_2) = (10111)_2$ and $L_4((11110)_2) =
(11110)_2$.
\end{Example}

\begin{Prop}\label{prop1}
Let $b$ be a positive integer greater than 1 and   $p$ be  a
positive integer greater than $b$ with the radix $b$ expansion
$(p_{m+1}p_m\cdots p_1)_b$. Suppose $k$ is a positive integer less
than p and there exists $m+1$ nonnegative integers
$k_{m+1},k_m,\cdots,k_1$ less than $b$ such that $k =
\sum_{i=1}^{m+1}k_ib^{i-1}$. We still denote it by
$(k_{m+1}k_m\cdots k_1)_b$ and $k_{m+1}$ may be zero.

If $k_{m+1} < p_{m+1}$,  one has
\begin{equation}
L_l(k) = (k_{m+1}k_{\sigma^l(m)}k_{\sigma^l(m-1)}\cdots k_{\sigma^l(1)})_b < p.
\end{equation}
Furthermore, if $k \in \mathbb{Z}_p^*$, then
\begin{equation}
L_{-l}(L_l(k)) = k
\end{equation}
in $\mathbb{Z}_p^*$.

\end{Prop}

\begin{proof}
Since $ \sum_{i=1}^m k_{\sigma^l(i)}b^{i-1} < b^m $ and $ k_{m+1} <
p_{m+1}$, one has
$$ L_l(k) = k_{m+1}b^m + \sum_{i=1}^m k_{\sigma^l(i)}b^{i-1} < k_{m+1}b^m+b^m \leq p_{m+1}b^m \leq p.$$
This implies that
$$L_l(k)\equiv (k_{m+1}k_{\sigma^l(m)}k_{\sigma^l(m-1)}\cdots k_{\sigma^l(1)})_b  \mbox{ (mod } p) .$$
Hence,
$$L_{-l}(L_l(k)) = L_{-l}((k_{m+1}k_{\sigma^l(m)}k_{\sigma^l(m-1)}\cdots k_{\sigma^l(1)})_b) =(k_{m+1}k_m\cdots k_1)$$
in $\mathbb{Z}_p^*$.
\end{proof}

\begin{Rmk}
In general, $L_{-l}(L_l(k))$ may not be equal to $k$ in
$\mathbb{Z}_p^*$ even if $k < p$. For instance, if $b=10$, $p=239$,
$l=1$ and $k=235$. Then
$$ L_{-1}(L_1(235)) = L_{-1}( 253 \mbox{ mod } 239 ) = L_{-1}(014) = 41$$
in $\mathbb{Z}_{239}^*$.

\end{Rmk}
 Now, we discuss our solution of the access control in hierarchy. In our method,
CA is a super security class  which should  choose a secrete integer
$h_i$ and a public integer $l_i$ for each secure-filter $f_i$. Let
$L_l(x)$ be a cyclic $l$-shift operator. Then the secure-filter in
our proposed method is given by
\begin{equation}
f_i(x) = (x-h_i)(x-a_{i,1})\cdots(x-a_{i,N_i}) +L_{l_i}(K_i)
\end{equation}
where $a_{i,j}$ plays the same role as  $\tilde{A}(n_jP_i)$ in
Jeng-Wang method.

\begin{description}
\item[Method 1:] (Revised Jeng-Wang method) \\
\begin{description}
\item[Initialization phase ]- \\
    \begin{enumerate}
    \item CA randomly chooses a small positive integer $d$ greater than 1, a large prime $p$, and $a,b \in
      \mathbb{Z}_p^*$ satisfying that $4a^3 + 27b^2
      \textrm{ mod } p \neq 0$ Let $E_p(a,b)$ be an
      elliptic curve $y^2=x^3+ax+b \mbox{ (mod }p)$ over
      a Galois field $GF(p)$ containing a set of points
    $(x,y)$ with $x,y \in \mathbb{Z}^*_p$ and a point
    $O$ at infinity. CA selects  a base point $G$ of
    $E_p(a,b)$  whose order is a  large prime $q$.
    \item CA selects an algorithm $\tilde{A} :(x,y) \rightarrow v$, for representing a point $(x,y)$ on $E_p(a,b)$ as a real number $v$.
    \item CA publishes $(d,p,q,\tilde{A},E,G)$.
    \end{enumerate}

\item[Key generation phase ] - \\
    \begin{enumerate}
    \item CA determines its security key $n_{ca} \in \mathbb{Z}_q^*$ and publishes the public key $P_{ca} =  n_{ca}G$.
    \item Each security class $u_i$ chooses a secret key $n_i \in
     \mathbb{Z}_q^*$ , a public key $P_i =n_iG$ and  its
          encryption key $K_i$ such that the leading
          coefficient of the
          radix $d$ expansion of $K_i$ is less than the
     leading coefficient of the radix $d$ expansion of
     $p$ and all coefficients of the radix $d$ expansion
     of $K_i$ are not all equal.
    \item $u_i$ chooses a random integer $k$ and encrypts $(K_i,n_i)$ as a ciphertext $\{kG, (K_i,n_i)+kP_{ca} \}$, and sends this ciphertext to CA.
    \item CA decrypts $\{kG, (K_i, n_i)+kP_{ca} \}$ to obtain $(K_i, n_i)$ by the equation
    \begin{equation}
    (K_i,n_i) = ((K_i,n_i)+kP_{ca}) - n_{ca}(kG).
    \end{equation}
    \item CA randomly choose a secret integer $h_i$ and a public integer $l_i$.
    \item CA generates the secure-filter
    \begin{equation}
    \begin{array}{rl}
    f_i(x)  = & (x-h_i)(x -\tilde{A}(n_1P_i)) \cdots(x-\tilde{A}(n_{N_i}P_i)) +L_{l_i}(K_i) \cr
            = & \prod_{u_j \in S_i} (x-\tilde{A}(n_jP_i)) +L_{l_i}(K_i)
    \end{array}
    \end{equation}
    and publishes it.
    \end{enumerate}

\item[key derivation phase] -

When $u_j$ is a predecessor of $u_i$, $u_j$ can use the
secure key $n_j$, the public key $P_i$ and the public
secure-filter $f_i$ to derive $L_{l_i}(K_i) =
f_i(\tilde{A}(n_jP_i))$. Finally, $u_j$ obtain  the
cryptographic key $K_i = L_{-l_i}(L_{l_i}(K_i))$ by
Proposition \ref{prop1}.
\end{description}

\end{description}

\begin{figure}
{\centering
\includegraphics[width=.5\textwidth]{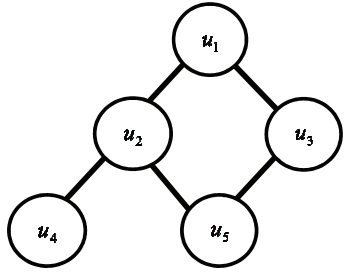}
\caption{ A hierarchy $(S,\leq)$. }\label{hierarchy_1}
}
\end{figure}

\begin{Example}
In Fig. \ref{hierarchy_1}, the user set has 5 security classes
denoted as $S=\{u_1,\cdots,u_5\}$. The secure-filter of $u_i$ in our
method forms as
\begin{equation}
f_i(x) = (x-h_i)\prod_{u_i <u_j} (x-\tilde{A}(n_jP_i)) + L_{l_i}(K_i).
\end{equation}
For simplicity, we denote $\tilde{A}(n_jP_i)$ as $a_{i,j}$ here and
after. All secure-filters in the hierarchy $S$ are
\begin{equation}
\begin{array}{ll}
u_1 : &  f_1(x) = 0 \cr
u_2 : & f_2(x)=(x- h_2)(x-a_{2,1})+L_{l_2}(K_2) \cr
u_3 : & f_3(x) = (x-h_3)(x-a_{3,1})+L_{l_3}(K_3) \cr
u_4 : & f_4(x) = (x-h_4)(x-a_{4,1})(x-a_{4,2})+L_{l_4}(K_4) \cr
u_5 : & f_5(x) = (x-h_5)(x-a_{5,1})(x-a_{5,2})(x-a_{5,3}) + L_{l_5}(K_5)
\end{array}
\end{equation}

\end{Example}

\subsection{ Security analysis}

When a new predecessor $u_N$ of $u_i$ is added in hierarchy, the
secure-filter of $u_i$,
 \begin{equation} f_i(x) =
(x-h_i)(x-a_{i,1})\cdots(x-a_{i,N_i}) +L_{l_i}(K_i),
\end{equation}
becomes
\begin{equation}
\tilde{f}_i(x) = (x-\tilde{h}_i)(x-a_{i,1})\cdots(x-a_{i,N_i}) +L_{\tilde{l}_i}(K_i).
\end{equation}

Because $h_i, \tilde{h}_i, l_i$ and $\tilde{l}_i$ can randomly be
chosen by CA, we may assume $h_i \neq \tilde{h_i}$ and $L_{l_i}(K_i)
\neq L_{\tilde{l}_i}(K_i)$. Therefore, $a_{i,1} \cdots a_{i,N_i}$
are not the solutions of $\tilde{f}_i(x)-f_i(x)=0$ and Lin-Hsu
attack fails. By comparing the coefficients of $x_{N_i}, x_{N_i-1}$
of $f_i(x)$ and the coefficients of $x_{N_i+1}, x_{N_i}$ of
$\tilde{f}_i(x)$, one has
\begin{equation}\label{coeff_x_N}
\left \{ \begin{array}{ll}
 \alpha_{N_i} &= a_{i,1}+a_{i,2}+\cdots+a_{i,N_i}+h_i \cr
 \beta_{N_i+1} &= a_{i,1}+a_{i,2}+\cdots+ a_{i,N_i}+a_{i,N}+\tilde{h}_i
\end{array} \right.
\end{equation}
and
\begin{equation}\label{coeff_x_N_1}
\left \{ \begin{array}{ll}
\alpha_{N_i-1} = \prod_{s\neq t} a_{i,s}a_{i,t} + h_i(\sum_s a_{i,s}) \cr\cr
\beta_{N_i} = \prod_{s \neq t} a_{i,s}a_{i,t} + (a_{i,N}+\tilde{h}_i)(\sum_s a_{i,s})+a_{i,N}\tilde{h}_i
\end{array}\right.
\end{equation}
where $1 \leq s,t \leq N_i$.

 Eq.(\ref{coeff_x_N}) and Eq.(\ref{coeff_x_N_1}) imply
\begin{equation}
\left \{
\begin{array}{ll}
\beta_{N_i+1} - \alpha_{N_i} & = a_{i,N}+\tilde{h}_i - h_i                 \cr
\beta_{N_i} - \alpha_{N_i-1} & = (a_{i,N}+\tilde{h}_i - h_i)(\sum a_{i,s}) + a_{i,N}\tilde{h}_i.
\end{array}  \right.
\end{equation}
Consequently,
\begin{equation}
\beta_{N_i} - \alpha_{N_i-1} = (\beta_{N_i+1} -\alpha_{N_i})(\alpha_{N_i} - h_i) + a_{i,N}\tilde{h}_i,
\end{equation}
that is,
\begin{equation}\label{main_eq1}
a_{i,N}\tilde{h}_i - (\beta_{N_i+1} - \alpha_{N_i})h_i = (\beta_{N_i} - \alpha_{N_i-1}) - (\beta_{N_i+1} - \alpha_{N_i})\alpha_{N_i}.
\end{equation}
Eq.(\ref{main_eq1}) is a nonlinear equation with indeterminates
$a_{i,N}, h_i$ and $\tilde{h}_i$ and it is hard to solve in
$\mathbb{Z}_p^*$.

\subsection{Dynamic key management} An efficient method for updating
the secure-filter is the main task in dynamic access control of our
proposed methods. Before   discussing the problems of dynamic key
management, we introduce some  properties for our secure-filters.

\begin{Thm}\label{thm1}
Let \begin{equation}f(x)=(x-a_1)(x-a_2)\cdots (x - a_n) + K_1
\end{equation} and
\begin{equation}\label{poly_g}
g(x)=(x-a_1)(x-a_2)\cdots(x-a_n)(x-a_{n+1}) + K_2
\end{equation}  be two
polynomials in $\mathbb{Z}_p^*$. Immediately, we have $f(x) =
\frac{g(x)-K_2}{(x-a_{n+1})} + K_1$ or $g(x) = (x-a_{n+1})(f(x)-K_1)
+ K_2$.

If $f(x)=\sum_{i=0}^n \alpha_i x^i $ and $g(x) =\sum_{i=0}^{n-1}
\beta_i x^i$, then we have

\begin{equation}\label{insert_formula}
\beta_j = \left \{ \begin{array}{ll}
1 & \mbox{ if } j=n+1 \cr
\alpha_{j-1} - a_{n+1}\alpha_j & \mbox{ if } j \in \{2,3,\cdots,n\} \cr
\alpha_0 - a_{n+1}\alpha_1 - K_1 & \mbox{ if } j=1 \cr
 (K_2 - K_1) -a_{n+1}\alpha_0  & \mbox{ if } j = 0
\end{array} \right.
\end{equation}
or
\begin{equation}\label{delete_formula}
\alpha_j = \left \{ \begin{array}{ll}
1 & \mbox{ if } j= n \cr
\beta_{j+1} + a_{n+1}\alpha_{j+1} & \mbox{ if } j \in \{1,2,\cdots,n-1\} \cr
\beta_1 + a_{n+1}\alpha_1 +K1 & \mbox{ if } j=0
\end{array}\right.
\end{equation}

\end{Thm}
\begin{proof}
One has
\begin{equation}\label{poly_g_2}
\begin{array}{rl}
g(x)  = &(x-a_{n+1}) \prod_{i=0}^n(x-a_i) + K_2 \cr
     =& (x-a_{n+1})(f(x)-K_1) + K_2 \cr
     =& (x-a_{n+1})(\alpha_nx^n + \alpha_{n-1}x^{n-1} \cdots + \alpha_0 - K_1) + K_2 \cr
     =& \alpha_nx^{n+1} + (\alpha_{n-1} - a_{n+1}\alpha_n)x^n + \cdots + (\alpha_1  - a_{n+1}\alpha_2)x^2 \cr
      &+ (\alpha_0 - K_1 -a_{n+1}\alpha_1)x+K_2-a_{n+1}(\alpha_0 - K_1).
\end{array}
\end{equation}
Comparing the coefficients of each terms of $f(x)$ and $g(x)$, we
obtain Eq.(\ref{insert_formula}) and (\ref{delete_formula}).
\end{proof}

\begin{Cor}\label{Cor1}
Let $f(x) = (x-h_1)(x-a_1)\cdots(x-a_n) +K_1$ be a polynomial of
degree $n+1$ in $\mathbb{Z}_p^*$ and $\alpha_i$ be the coefficient
of the term $x^i$ of the general form of $f(x)$. Hence, $f(x) =
\sum_{i=0}^{n+1} \alpha_ix^i$.

If $g(x) = (x-h_2)(x-a_1)\cdots (x-a_n)(x-a_{n+1}) +K_2$ is a
polynomial of degree $n+2$ in $\mathbb{Z}_p^*$ then the coefficient
$\beta_j$ of $x^j$ of $g(x)$ can be computed by the formulas in
Theorem \ref{thm1}.
\end{Cor}
\begin{proof}
Since
$$g(x)=(x-a_{n+1})\left \{ (x-h_2) \left [ \frac{f(x)-K_1}{(x-h_1)} +0  \right ] +0  \right \} +K_2.$$
Hence, the coefficients of $G_1(x)=\frac{f(x)-K_1}{(x-h_1)} + 0 $
can be estimated by Eq.(\ref{delete_formula}),  and the coefficients
of $G_2(x)=(x-h_2)(G_1(x)-0) + 0)$ can be estimated by
Eq.(\ref{insert_formula}). Finally, we obtain  the coefficients of
$g(x)=(x-a_{n+1})(G_2(x)-0) + K_2$ by the Eq.(\ref{insert_formula}).
\end{proof}

\begin{Cor}\label{Cor2}
Under the  same assumptions in Corollary \ref{Cor1}. If
$h(x)=(x-h_3)(x-a_1)\cdots(x-a_{n-1})+K_3$ is a polynomial of degree
$n$, then the coefficient of $x^j$ of $h(x)$ can be obtained by
Theorem \ref{thm1}.
\end{Cor}
\begin{proof}
Since
$$ h(x)=(x-h_3) \left \{ \frac{1}{(x-h_1)} \left [ \frac{f(x)-K_1}{(x-a_n)} +0 \right ]+ 0 \right \} + K_3,$$
the coefficients of $H_1(x)=\frac{f(x)-K_1}{x-a_n}$ can be evaluated
by Eq. (\ref{delete_formula}), the coefficients of $H_2(x) =
\frac{H_1(x)-0}{(x-h_1)}$ can be computed by Eq.
(\ref{delete_formula}). And we obtain the coefficients of
$h(x)=(x-h_3)(H_2(x)-0) +K_3$ by Eq. (\ref{insert_formula}).
\end{proof}

\begin{description}
\item[Inserting new security class:] ~\\
Assume that a new security class $u_r$ is inserted into the
hierarchy such that $u_i < u_r < u_j$.
\begin{enumerate}
\item Determine the secret information $u_r, K_r, n_r, h_r$, the public information $P_r , l_r$ and the
     public secure-filter $f_r(x)$ by the steps 2-6 in key
     generation phase of our proposed method.

\item For each $u_s < u_r$ (i.e. the strictly successor of $u_r$)
\begin{enumerate}
\item Update the secret number $h_s$ and the public number $l_s$.
\item Update the secure-filter $f_s(x)$ of $u_s$ by Corollary \ref{Cor1}.

\end{enumerate}

\end{enumerate}

\item[Removing existing security classes :] ~\\
Assume that an existing member $u_r$ is   removed from a
hierarchy $(S,\leq)$ and $S_r = \{ u_s | u_r \leq u_s \}$ is the
set of all strictly predecessors of $u_r$.

\begin{enumerate}
\item For each strictly predecessor $u_s$ of $u_r$,
\begin{enumerate}
\item CA chooses a new secret number $h_s$ and a new
    public number $l_s$ randomly.
\item Update the secret filter $f_s$ by Corollary \ref{Cor2}.

\end{enumerate}
\item Delete the security class $u_r$ from the hierarchy $(S,\leq)$ and discard
the secret and public information of $u_r$.

\end{enumerate}

\end{description}

\begin{figure}
{\centering
\includegraphics[width=.5\textwidth]{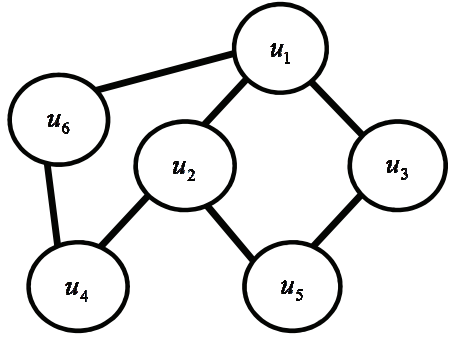}
\caption{ A new security class $u_6$ is inserted into the hierarchy $(S,\leq)$. }\label{hierarchy_2}
}
\end{figure}

\begin{Example}
In Fig. \ref{hierarchy_2}, a new security class $u_6$ is inserted
into the hierarchy $S$ in Fig. \ref{hierarchy_1} such that $u_4 <
u_6 < u_1$. The secure-filter of $u_6$ is
\begin{equation}
f_6(x) = (x-h_6)(x-a_{6,1}) + L_{l_6}(K_6).
\end{equation}
Since $u_4$ is the only strictly successor of $u_6$, CA only needs
to update the secure-filter of $u_4$. Assume that CA chooses a new
secrete number $\tilde{h}_4$ and a public number $\tilde{l}_4$
randomly.  After $u_6$ is inserted, the secure-filter of
$u_4$becomes
\begin{equation}\label{f_4}
f_4(x) = (x - \tilde{h}_4)(x-a_{4,1})(x-a_{4,2})(x-a_{4,6}) + L_{\tilde{l}_4}(K_4).
\end{equation}
Now, let us compute the coefficients of this new secure-filter in
Eq.(\ref{f_4}) from the original secure-filter,$(x -
h_4)(x-a_{4,1})(x-a_{4,2}) + L_{l_4}(K_4)$, by Eqs.
(\ref{insert_formula}) and (\ref{delete_formula}). Assume that
\begin{equation}
 (x - h_4)(x-a_{4,2})(x-a_{4,1})+L_{l_4}(K_4) = \sum_{i=0}^3
\alpha_i^1 x^i,
\end{equation}

\begin{enumerate}
\item Using Eq. (\ref{delete_formula}), the coefficients of the indeterminate $x$ of
\begin{equation} (x-a_{4,1})(x-a_{4,2})+0 = \sum_{i=0}^2 \alpha_i^2x^i \end{equation}
can be obtained by
\begin{equation}
\begin{array}{ll}
\alpha^2_2 & = 1 \cr
\alpha^2_1 &  = -(a_{4,1}+a_{4,2})= \alpha^1_2 + h_4 \alpha^2_2 \cr
\alpha^2_0 &= a_{4,1}a_{4,2}  = \alpha^1_1 + 0  - h_4\alpha^2_1
\end{array}
\end{equation}

\item Using Eg. (\ref{insert_formula}), CA computes the coefficients of
\begin{equation}
(x-\tilde{h}_4)(x-a_{4,1})(x-a_{4,2}) + 0 = \sum_{i=0}^3
\alpha^3_ix^i; \end{equation}

\begin{equation}
\begin{array}{ll}
\alpha^3_3 & = 1 \cr
\alpha^3_2 &  = -(a_{4,1}+a_{4,2}) - \tilde{h}_4 = \alpha^2_1 - \tilde{h}_4 \alpha^2_2 \cr
\alpha^3_1 &  = a_{4,1}a_{4,2} + \tilde{h}_4(a_{4,1}+a_{4,2}) = \alpha^2_0   - \tilde{h}_4\alpha^2_1 \cr
\alpha^3_0 &  = -\tilde{h}_4a_{4,1}a_{4,2} = -\tilde{h}_4\alpha^2_0 + 0.
\end{array}
\end{equation}
\item
Using Eq. (\ref{insert_formula}) again, the coefficients of
\begin{equation}
(x-\tilde{h}_4)(x - a_{4,1})(x-a_{4,2})(x-a_{4,6})+L_{\tilde{l}_4}(K_4) =
\sum_{i=0}^4 \alpha^4_ix^i\end{equation}
are
\begin{equation}
\begin{array}{ll}
\alpha^4_4 & = 1 \cr
\alpha^4_3 &  = -(a_{4,1}+a_{4,2}+\tilde{h}_4) - a_{4,6} = \alpha^3_2 - a_{4,6}\alpha^3_3\cr
\alpha^4_2 &  = (a_{4,1}a_{4,2}+ \tilde{h}_4a_{4,1} +\tilde{h}_4 a_{4,2}) -a_{4,6}(-\tilde{h}_4-a_{4,1}-a_{4,2}) = \alpha^3_1 - a_{4,6}\alpha^3_2\cr
\alpha^4_1 & = -\tilde{h}_4a_{4,1}a_{4,2} - a_{4,6}(a_{4,1}a_{4,2})+\tilde{h}_4a_{4,1} + \tilde{h}_4a_{4,2} ) =  \alpha^3_0 - a_{4,6}\alpha^3_1  \cr
\alpha^4_0 & = a_{4,1}a_{4,2}a_{4,6}\tilde{h}_4 + L_{\tilde{l}_4}(K_4)    = a_{4,6}\alpha^3_0 + L_{\tilde{l}_4}(L_{-l_4}(L_{l_4}(K_4))).
\end{array}
\end{equation}

\item
Finally, CA updates the secure-filter of $u_4$ as the coefficients
$\{\alpha^4_0, \alpha^4_1, \cdots, \alpha^4_4 \}$.
\end{enumerate}

\end{Example}

\begin{figure}
{\centering
\includegraphics[width=.5\textwidth]{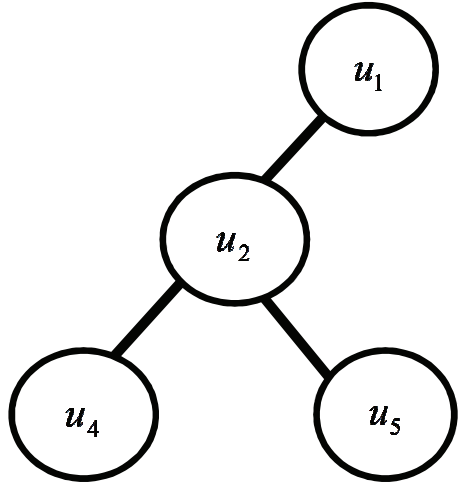}
\caption{ An existing member $u_2$ is removed from the hierarchy $(S,\leq)$. }\label{hierarchy_3}
}
\end{figure}

\begin{Example}
In Fig. \ref{hierarchy_3}, let $u_3$ be removed from the hierarchy
$(S, \leq)$ so that the relationship $u_5 < u_3 < u_1$  breaks up.
CA removes all information related to $u_3$ and choose four random
numbers $\tilde{h}_4, \tilde{l}_4$, $\tilde{h}_5$ and $\tilde{l}_5$.
CA replaces the secure-filter of $u_5$ by

\begin{equation}
f_5(x) = (x-\tilde{h}_5)(x-a_{5,1})(x-a_{5,2}) +L_{\tilde{l}_5}(K_5).
\end{equation}

Assume that $(x-h_5)(x-a_{5,1})(x-a_{5,2})(x-a_{5,3}) + L_{l_5}(K_5)
= \sum_{i=0}^4{\beta^1_i}x^i$.

\begin{enumerate}
\item  The coefficients of $(x-a_{5,1})(x-a_{5,2})(x-a_{5,3})+0 =
\sum_{i=0}^3 \beta^2_ix^i$ are
\begin{equation}
\begin{array}{ll}
\beta^2_3 & = 1 \cr
\beta^2_2 &= \beta^1_3 + h_5\beta^2_3   \cr
\beta^2_1 & = \beta^1_2 + h_5\beta^2_2    \cr
\beta^2_0 & = \beta^1_1 -h_5(\beta^2_1) + 0
\end{array}
\end{equation}

\item The coefficients of $(x-a_{5,1})(x-a_{5,2})+0 =
\sum_{i=0}^2\beta^3_ix^i$ are
\begin{equation}
\begin{array}{ll}
\beta^3_2 & = 1 \cr
\beta^3_1 &= \beta^2_2 + a_{5,3}\beta^3_2  \cr
\beta^3_0 & = \beta^2_1 + a_{5,3}\beta^3_1.
 \end{array}
\end{equation}

\item The coefficients of $(x-\tilde{h}_5)(x-a_{5,1})(x-a_{5,2}) +
L_{\tilde{l}_5}(K_5)$ are
\begin{equation}
\begin{array}{ll}
\beta^4_3 & = 1 \cr
\beta^4_2 &= \beta^3_1 - h_5\beta^3_2  \cr
\beta^4_1 & = \beta^3_0 -0 - h_5\beta^3_1   \cr
\beta^4_0 & =   -h_5(\beta^3_0 - 0) + L_{\tilde{l}_5}(L_{-l_5}(L_{l_5}(K_5))).
\end{array}
\end{equation}

\item Then, CA publishes these coefficients $\beta^4_0, \cdots
\beta^4_3$.
\end{enumerate}
\end{Example}

\section{Conclusion}

Our proposed method in this paper is a secure solution of access
control in hierarchy and the scheme of dynamic key management is
very simple and efficient. Furthermore, Wu et.al.'s method becomes a
secure method under our secure-filter.

\begin{description}

\item[Method 2:] (Revised  Wu et. al. method) \\
    \begin{description}
    \item[Key generation phase] - \\
        \begin{enumerate}
        \item Choose distinct secret integer $s_i \in \mathbb{Z}_p$, $ 0 \leq i \leq n-1$.
        \item CA chooses a small integer $d$ greater than 1 and publishes $d$.
        \item For each $u_i \in S$ do step 3-5
        \item Determine the set $\Sigma_i = \{ g_i^{s_j} | u_j \in S_j \}$.
        \item CA randomly choose a secret integer $h_i$ and a public integer $l_i$.
        \item Determine  the encryption key $K_i$ such that the leading coefficient of the
          radix $d$ expansion of $K_i$ is less than the
     leading coefficient of the radix $d$ expansion of
     $p$ and all coefficients of the radix $d$ expansion
     of $K_i$ are not all equal.
        \item Set
        \begin{equation}
        \begin{array}{rl}
        f_i(x) =& (x - h_i)(x - g_i^{s_0})(x -g_i^{s_1})\cdots(x-g_i^{s_{N_i}}) +L_{l_i}(K_i) \cr
               = &(x-h_i)\cdot\prod_{a \in \Sigma_i} (x-a) +L_{l_i}(K_i)
        \end{array}
        \end{equation}
        \item $u_i$ keeps $s_i$ and $K_i$ private, and makes $(f_i, g_i, l_i)$  public.
        \end{enumerate}

    \item[Key derivation phase ] -
    Consequently, the predecessor $u_j$ of $u_i$ can compute
    $g_i^{s_j}$ from his secure number $s_j$ and the public
    number $g_i$. Then the user $u_j$  obtain $L_{l_i}(K_i)$
    by $f_i(g_i^{s_j})$ as well as the cryptographic key
    $K_i$ by a cyclic  $-l_i$-shift operator$L_{-l_i}$, $K_i
    = L_{-l_i}(L_{l_i}(K_i))$.

\end{description}
\end{description}

\section*{Acknowledgements}
This paper is partially supported by NSC, Taiwan.


\end{document}